\documentclass[preprint,12pt]{elsarticle}
\usepackage{amssymb,amsmath,amsthm}
\usepackage[utf8x]{inputenc}
\usepackage{hyperref}
\usepackage{tikz}
\usetikzlibrary{arrows,calc,petri}
\tikzset{%
   every place/.style={thick, minimum size=8mm},
   every transition/.style={thick, minimum size=6mm},
   pre/.style={<-,shorten <=1pt,>=stealth',thick},
   post/.style={->,shorten >=1pt,>=stealth',thick},
   round/.style={rounded corners=5pt},
   fire/.style={transition,fill=yellow},
   posreg/.style={->,shorten >=1pt,>=stealth',very thick,green},
   negreg/.style={-|,shorten >=1pt,>=stealth',very thick,red}
}
\usepackage{listings}
\usepackage{xspace}
\usepackage{color}
\usepackage{algorithm}
\usepackage{algpseudocode}
\usepackage{booktabs}
\newcommand{\hl}[1]{#1}

\newif\ifcomments%
\commentstrue%

\newcommand{\lra}{\longrightarrow}
\renewcommand{\H}{\ensuremath{\mathcal{H}}\xspace}

\DeclareMathOperator{\on}{on}
\DeclareMathOperator{\off}{off}

\newtheorem{theorem}{Theorem}[section]
\newtheorem{corollary}[theorem]{Corollary}

\newtheorem{proposition}[theorem]{Proposition}
\newtheorem{definition}{Definition}

\journal{Journal of Theoretical Biology}

\begin{document}

\begin{frontmatter}

\title{\hl{Graphical Requirements} for Multistationarity in Reaction Networks and their Verification in BioModels\footnote{Article submitted to a special issue \emph{in memoriam} of René Thomas.}}

\author{Adrien Baudier}
\author{François Fages}
\ead{francois.fages@inria.fr}
\author{Sylvain Soliman}
\ead{sylvain.soliman@inria.fr}

\address{Inria Saclay Île-de-France, Palaiseau, France}

\begin{abstract}
\hl{Thomas's necessary conditions for the existence of multiple steady states in gene networks
have been proved by Soul\'e  with high generality for dynamical systems defined by differential equations.
When applied to (protein) reaction networks however, 
those conditions do not provide information  
since they are trivially satisfied
as soon as there is a bimolecular or a reversible reaction.
Refined graphical requirements have been proposed to deal with such cases.
In this paper, we present for the first time a graph rewriting algorithm for checking the refined conditions
given by Soliman, and evaluate its practical performance 
by applying it systematically to the curated branch of the BioModels repository.
This algorithm analyzes all reaction networks (of size up to 430 species) in less than 0.05 second per network,
and permits to conclude to the absence of multistationarity in 160 networks over 506.
The short computation times obtained in this graphical approach
are in sharp contrast to the Jacobian-based symbolic computation approach.
We also discuss the case of one extra graphical condition by arc rewiring that allows us to conclude on 20 more networks of this benchmark but with a high computational cost.
Finally, we study with some details the case of phosphorylation cycles and MAPK signalling models
which show the importance of modelling the intermediate complexations with the enzymes in order to correctly analyze the multistationarity capabilities of such biochemical reaction networks.
}
\end{abstract}


\end{frontmatter}


\section{Introduction}%
\label{sec:intro}

\hl{The wide variety of cells in a multicellular organism show that cells with identical copies of DNA may differentiate in different cell types.
In the late 40's,  Max Delbruck at Caltech suggested that each type of cell could correspond to a distinct steady state
in the dynamics of their shared gene expression network.
In order to analyze such large networks, Ren\'e Thomas conjectured in 1980 that the existence of a positive (resp.~negative) feedback loop 
was a necessary condition for multistationarity (resp.~sustained oscillations)}~\cite{Thomas81sss}.
\hl{Those conjectures were later proved in various formalisms (Boolean or discrete transition systems, differential equations)
with various degrees of generality.
In 2003, Christophe Soul\'e finally proved Thomas's necessary condition for multistationarity with full generality for dynamical systems defined by differential equations}~\cite{Soule03complexus}.

\hl{In his mathematical formalization of the conjecture, Soul\'e considers a differentiable mapping $F$ from a finite dimensional real vector space to itself,
and for each point $a$, the directed graph $G(a)$ where the arcs are the non-zero entries of the Jacobian matrix of $F$, labeled by their sign.
He shows that if $F$ has at least two non-degenerate zeroes, there exists $a$ such that $G(a)$ has a positive circuit.
}

\hl{When applied to (protein) reaction networks however, Thomas's necessary condition for multistationarity fails short since it is trivially satisfied
as soon as there exists either a bimolecular or a reversible reaction.
}
Indeed, 
a bimolecular reaction such as a complexation reaction immediately creates a mutual inhibition between the two reactants, i.e.~a
positive circuit, and a reversible reaction produces a mutual activation, i.e.~again
a positive circuit, making Thomas's necessary condition always true in those networks.

Nevertheless, reaction models are widespread in computational systems biology and it would be very desirable to be able to predict the absence of multistationarity
by systematically checking such conditions with efficient algorithms.
For instance, the BioModels database\footnote{\url{http://biomodels.net/}}~\cite{CLN13issb} is a
repository of more than 600 hand-curated models written in the Systems Biology Markup
Language (SBML)~\cite{Hucka08sbml2} mostly with reaction rules, over several tenths or hundred of molecular species.
There are hundreds more models in the non-curated branch, and thousands of models imported from metabolic networks databases with even larger numbers of reactions and species.

\hl{%
Soul\'e's proof, as most preceding and following proofs, uses the fact that the existence of multiple steady states 
implies a non-injectivity property which is shown to be equivalent
to a determinant being zero for some values of reaction rate constants.
One approach, called the Jacobian approach, is thus to use symbolic computation methods to directly compute the roots of that determinant.
If it is non-zero, one can conclude to the absence of multistationarity.
This is the approach taken by Feliu and Wiuf in}~\cite{FW13bi}.
\hl{%
Interestingly, they evaluated their algorithm, implemented in Maple 16, on the curated branch of BioModels (323 networks in their case),
showing that 31,6\% were injective and that only 8,3\% of the networks of this benchmark caused memory overflow by that method.
On the sequences of $r$ phosphorylation cycles of}~\cite{WS08jmb},
\hl{%
they could check non-injectivity up to $r=17$ cycles in 1200 seconds.

In this paper, we follow the alternative graphical approach to multistationarity analyses.
We describe a graph rewriting algorithm which deals with sequences of $r=1000$ phosphorylation cycles in a second,
and analyzes the curated branch of BioModels (506 networks in our case)
with a maximum computation time of 50 milliseconds per network (including large networks of size up to 430 species), 
while concluding to the non existence of multiple steady states in 160 networks of size up to 54 species in that benchmark, i.e.~with a similar ratio of 31.6\% of results concluding to non-multistationarity.

This algorithm is based on a refinement of the graphical requirements of Soul\'e}~\cite{Soule03complexus}\hl{%
given by the third author in}~\cite{Soliman13bmb} \hl{as a necessary condition for the existence of multiple steady states in (biochemical) reaction networks.
Similar graphical requirements have also been given in}~\cite{BC10aam}\hl{ without restriction to mass-action law kinetics,
but to our knowledge, it is the first time that they are implemented and evaluated systematically in model repositories.
}
For instance, we are not aware of similar evaluations obtained
with the Chemical Reaction Network Toolbox\footnote{\url{https://crnt.osu.edu/toolbox-history-and-explanation}} for \hl{systematically} checking the 
graphical conditions for multistationarity of Feinberg's Chemical
Reaction Network Theory (CRNT)~\cite{Feinberg77crt,CF06siamjam}.

\hl{More specifically, we present a series of graph rewriting algorithms for checking the different graphical requirements} of~\cite{Soliman13bmb},
and \hl{analyze their practical performance} in the curated models of BioModels, in order to:
\begin{itemize}
   \item evaluate when the original condition of Thomas allows \hl{us} to rule out
      multistationarity;

   \item evaluate when the following three extra conditions given in~\cite{Soliman13bmb} become conclusive, namely:

      \begin{enumerate}

         \item the positive circuit must not come from twice the same reaction;

         \item the positive circuit must not come from a reaction and its
            reverse reaction;

         \item the positive circuit must not involve all species of a
            conservation law;

      \end{enumerate}

   \item evaluate when even stronger conditions based on the
      rewirings detailed in~\cite{Soule03complexus,Soliman13bmb} are necessary to conclude, namely 
      \begin{enumerate}
         \item by sign change of incoming arcs on a set of species,
         \item or by permuting the arcs to a set of target species.
      \end{enumerate}
\end{itemize}

For this study, we used our software modelling environment BIOCHAM\footnote{\url{http://lifeware.inria.fr/biocham4}}~\cite{biocham4,cfs06bi}
to load all models from the curated branch of BioModels,
improve their writing in SBML with well-formed reactions using the algorithm described in~\cite{FGS15tcs},
compute the conservation laws~\cite{Soliman12amb},
compute their influence multigraph labelled by the reactions~\cite{FMRS18tcbb,FS08fmsb}
and export the labelled multigraph in the Lemon library format\footnote{\url{http://lemon.cs.elte.hu/}}.
Then we used an implementation in C++ 
of the algorithm presented in this paper to search for positive circuits
with the different refined conditions on the labelled influence multigraph,
and evaluate their respective contributions for the analysis of multistationarity in BioModels. 
\hl{All the computation times obtained with this algorithm given in this paper were obtained on a standalone desktop Linux machine with an Intel Xeon 3.6 GHz processor\footnote{For the sake of reproducibility, our programs and data are available at \url{https://lifeware.inria.fr/wiki/Main/Software\#JTB18}}.}

The rest of this article is organized as follows. The next section presents the refined necessary conditions for multistationarity in reaction networks
\hl{described} in~\cite{Soliman13bmb} and detailed here with five levels of conditions.
The following section presents a graph rewriting algorithm for checking those conditions, and evaluates its computational complexity.
Section~\ref{sec:biomodels} \hl{shows the remarkable performance of this algorithm by applying it} systematically to the curated part of the model repository BioModels, 
including models out of reach of Jacobian-based symbolic computation methods,
\hl{and details the effect of the five levels of refined conditions in this benchmark}.
Section~\ref{futile} considers the models of double phosphorylation cycles of Wang and Sontag \cite{WS08jmb} and 
shows a very low quadratic empirical complexity of the graphical algorithm, again in sharp contrast to 
symbolic computation methods.
Section~\ref{sec:erk} focusses on model 270 of ERK signalling that contains 33 species and 42 reactions resulting in an influence multigraph of 126 arcs
with many positive and negative feedback loops, yet for which our graphical algorithm demonstrates the absence of multistationarity.
\hl{These examples illustrate} the importance of \hl{modelling} the intermediate complexes in enzymatic reactions
\hl{to obtain multiple steady states}, and show the sensitivity of both the dynamical properties of the models and of our graphical conditions to
the writing of enzymatic reactions with or without intermediate complexes.
We conclude on the \hl{remarkable} performance of the graphical approach to analyze multistationarity in reaction models of large size,
and on some perspectives to further improve our algorithm and generalize this approach.

\section{Necessary Condition for Multistationarity in Reaction Networks}

Let us consider a biochemical reaction system with $n$ species
$S_1,\dots,S_n$ and $m$ reactions $R_1,\dots,R_m$. Using notations from~\cite{Kaltenbach12arxiv} we write:
\[R_j = \sum_{i=1}^n y_{ij}S_i \lra \sum_{i=1}^n y'_{ij}S_i\]
The $y$ and $y'$ represent the stoichiometric coefficients of the reactants
and products of the reaction.
The rate law associated with reaction $R_j$ will be written $v_j$. This
defines a dynamical system, in the form of an Ordinary Differential Equation
(ODE): $\dot x = F(x)$
where $x_i$ is the concentration of species $S_i$ and
\[f_i(x) = \sum_j v_j(x)\cdot (y'_{ij} - y_{ij})\]

This kind of reaction-based system encompasses most of the systems biology
models developed nowadays and made available in model repositories like BioModels. 
In particular, SBML reaction models can be translated with our notations,
basically by splitting reversible reactions into forward and backward
reactions, and by including modifiers on both sides of the reaction.

Reaction systems are often graphically represented as a Petri-net, i.e., a bipartite
graph for species and reactions~\cite{Ivanova79kk,IT79kk}. Using the same
bipartite vertices but different arcs and labels, it is possible to represent
the Directed Species-Reaction (DSR) graph of Kaltenchbach~\cite{Kaltenbach12arxiv}.
\hl{This graph is a variant of the DSR graph} of~\cite{BC09cms,BC10aam} \hl{with different labels and no sign}. 
\hl{Here the arcs of the DSR graph are defined and identified by their label $\lambda$} as follows: \[ \lambda(S_i, R_j) =
\frac{\partial v_j}{\partial x_i}\qquad \lambda(R_j, S_i) = y'_{ij} - y_{ij}
\]
If $\lambda$ is zero, then there is no arc. $\lambda$ is extended to
paths (resp.\ subgraphs) as the product of the labels of all arcs in the path
(resp.\ subgraph). For a path $P$, we shall write $\lambda_{SR}(P)$ (resp.\
$\lambda_{RS}(P)$) for the product of labels considering only species to
reaction (resp.\ reaction to species) arcs.

Intuitively, $\lambda_{SR}$ represent\hl{s} the contribution of species to each
reaction rate, whereas $\lambda_{RS}$ describe\hl{s} the stoichiometric effect
of reactions on each species.
Fig.~\ref{fig:dsr} shows the DSR graph for the chemical reaction network
corresponding to the enzymatic reaction $S + E \rightleftarrows ES
\longrightarrow E + P$.

\begin{figure}[htb]
   \begin{center}
      \begin{tikzpicture}
         \matrix[row sep=9mm,column sep=14mm] {%
            & & \node[place] (E) {E}; & & &\\
            & & \node[transition] (r1) {$R_1$}; & & &\\
            & \node[place] (S) {S}; & & \node[place] (ES) {ES};
            & \node[transition] (r3) {$R_2$}; & \node[place] (P) {P};\\
            & & \node[transition] (r2) {$R_{-1}$}; & & &\\
         };

         \draw[pre] (r1) to[bend right] node[right] {$\frac{\partial v_1}{\partial E}$} (E);
         \draw[post] (r1) to[bend left] node[left] {-1} (E);
         \draw[pre] (r1)  to[bend right] node[left] {$\frac{\partial v_1}{\partial S}$} (S);
         \draw[post] (r1)  to[bend left] node[right] {-1} (S);
         \draw[post] (r1) -- node[above] {1} (ES);
         \draw[pre] (r2)  to[bend right] node[right] {$\frac{\partial
               v_{-1}}{\partial ES}$} (ES);
         \draw[post] (r2)  to[bend left] node[left] {-1} (ES);
         \draw[post] (r2) -- node[above] {1} (S);
         \draw[pre] (r3)  to[bend right] node[above] {$\frac{\partial
               v_2}{\partial ES}$} (ES);
         \draw[post] (r3)  to[bend left] node[below] {-1} (ES);
         \draw[post] (r3) -- node[above] {1} (P);
         \draw[round, post] (r2) -- node[above] {1} ++(-35mm,0) |- (E);
         \draw[round, post] (r3) --  node[right] {1} (r3 |- E) -- (E);
      \end{tikzpicture}
   \end{center}
   \caption{DSR graph of the enzymatic reaction:
      $S + E \rightleftarrows ES \lra E + P$.}\label{fig:dsr}
\end{figure}

\begin{definition}\cite{Kaltenbach12arxiv} 
A \emph{species hamiltonian hooping} of the DSR graph is a collection of
   cycles covering each of the species nodes exactly once.
\end{definition}

   The set of all species hamiltonian hoopings will be denoted by $\H$.
Thanks to the fact that $\lambda(H) = \lambda_{SR}(H) \lambda_{RS}(H)$,
Kaltenbach~\cite{Kaltenbach12arxiv} proposed to group all species hamiltonian
hoopings having the same species-to-reaction arcs using an equivalence
relation noted $\sim$, writing $[H] = \{H'\in\H\mid H'\sim H\}$ for the equivalence class of a hooping
$H$, and
$\H/\sim$ for the quotient set. 

\begin{theorem}[\cite{Kaltenbach12arxiv}]\label{thm:kaltenbach}
   \[\det(J) = \sum_{[H] \in \H/\sim}\Lambda([H]) \lambda_{SR}(H)
      \quad\text{ with }\quad
   \Lambda([H]) = \sum_{H' \in [H]}\sigma(H') \lambda_{RS}(H')\]
\end{theorem}

Considering Soulé's proof of Thomas's conjecture for dynamical
systems~\cite{Soule03complexus} and applying Thm.~\ref{thm:kaltenbach} to each
sub-DSR-graph corresponding to a principal minor of $-J$,~\cite{Soliman13bmb}
notes that a necessary condition for multistationarity is that some term of
the sum is negative. This again states the usual condition about the existence
of a positive cycle in the influence graph of $J$.

\hl{Now, the usual labelling of the arcs of the influence graph between molecular species by the sign of the Jacobian matrix coefficient
can be augmented to
contain not only the sign but also the reaction used for each arc. There is
thus an arc in this reaction-labelled influence multigraph for each
species-to-species path of length two in the DSR graph. This leads to a
one-to-one correspondence between hamiltonian hoopings of the
reaction-labelled influence multigraph and species hamiltonian hoopings of the DSR
graph. Fig.~}\ref{fig:labelled} \hl{illustrates this on the example of
Fig.~}\ref{fig:dsr}.

\begin{figure}[htb]
   \begin{center}
      \begin{tikzpicture}
         \matrix[row sep=3cm,column sep=2cm] {%
            & \node[place] (E) {E}; & & &\\
            \node[place] (S) {S}; & & \node[place] (ES) {ES};
            & & \node[place] (P) {P};\\
         };

         \draw[posreg,bend left] (E) to node[auto,black] {$R_1$} (ES);
         \draw[posreg] (S) to node[auto,black] {$R_1$} (ES);
         \draw[posreg,bend left] (ES) to node[auto,black] {$R_{-1}$} (E);
         \draw[posreg,bend left] (ES) to node[auto,black] {$R_{-1}$} (S);
         \draw[posreg] (ES) to node[auto,black] {$R_2$} (E);
         \draw[posreg] (ES) to node[auto,black] {$R_2$} (P);
         \draw[negreg,bend left] (E) to node[auto,black] {$R_1$} (S);
         \draw[negreg,bend left] (S) to node[auto,black] {$R_1$} (E);
      \end{tikzpicture}
   \end{center}
   \caption{Influence \hl{multigraph} associated to the Michaelis-Menten reaction
      system of the three reactions $S + E \rightleftarrows ES \lra E + P$. 
The influence arcs are
      labelled both by their sign, as usual and by the unique reaction from which they originate.
Note for instance that there are two positive
      arcs from \lstinline|ES| to \lstinline|E|. Negative self-loops are omitted for
      clarity}\label{fig:labelled}
\end{figure}
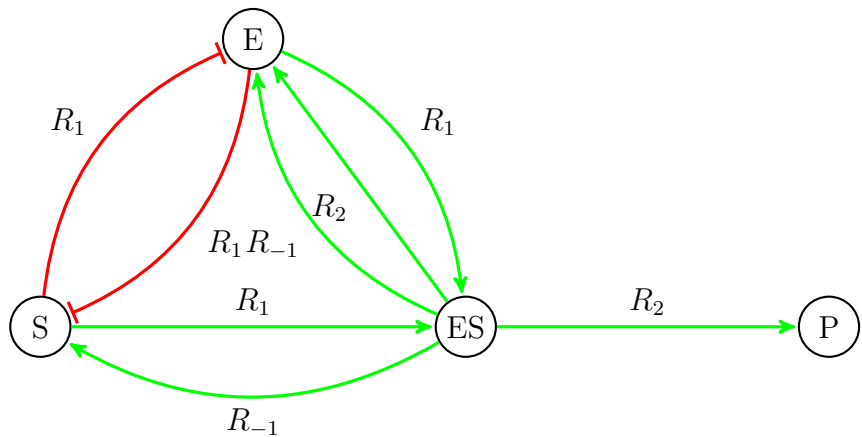

\hl{Let us denote by $|_H$ the} restriction of the reaction system to a species hooping $H$,
i.e.~the
   system where reactions $\{R_i \mid i \in I\}$ not appearing in $H$ are
   omitted.

\begin{theorem}[\cite{Soliman13bmb}]\label{thm:fullrank}
Let $F$ be any differentiable map from $\Omega$ to $\mathbb{R}^n$
   corresponding to a biochemical reaction system. If $\Omega$ is open and $F$
   has two nondegenerate zeroes in $\Omega$ then there exists some $a$ in
   $\Omega$ such that:
   \begin{enumerate}
      \item The reaction-labelled influence graph $G$ of $F$ at point $a$
         contains a positive circuit $C$;

      \item There exists a hooping $H$ in $G$, such that $C$ is subcycle of
         $H$ with $(Y' - Y)|_H$ of full rank.
   \end{enumerate}
\end{theorem}
\hl{This theorem gives several graphical requirements for multistationarity.} 

\begin{corollary}\label{cor:once}

   A necessary condition for the multistationarity of a biochemical reaction
   system is that there exists a positive cycle in its influence \hl{multigraph}, using
   \emph{at most once each reaction}.

\end{corollary}
This condition actually only requires the reaction-labelled influence \hl{multigraph}.
It is immediate to check that the mutual inhibition resulting from bimolecular
reactions---like that between \lstinline|E| and \lstinline|S| in our running
example---cannot fulfill these necessary conditions, since the
same reaction---$R_1$ in Fig.~\ref{fig:labelled}---will be repeated twice.

\begin{corollary}\label{cor:reversible}

   A necessary condition for the multistationarity of a biochemical reaction
   system is that there exists a positive cycle in its influence \hl{multigraph}, not
   using \emph{both forward and backward directions} of any reversible
   reaction.

\end{corollary}
The mutual activation resulting from reversible
reactions---like that between \lstinline|ES| and \lstinline|S| through $R_1$
and $R_{-1}$ in our running example---cannot thus fulfill these necessary conditions.
That condition is also a corollary of the conditions given in \cite{BC10aam}
since, in their setting,  reversible reactions give rise to a unique undirected edge.

Another information that can be extracted from the stoichiometry is the
(structural) conservation laws, i.e., P-invariants of the underlying Petri
net, or more simply the left kernel of the stoichiometry matrix. Finding all
the conservation laws of a biochemical model might be computationally
expensive, though in practice that does not seem to be the
case~\cite{Soliman12amb}.

\begin{corollary}\label{cor:pinv}

   A necessary condition for the multistationarity of a biochemical reaction
   system is that there exists a positive cycle in its influence \hl{multigraph}, not
   using \emph{all species involved in a conservation law}.

\end{corollary}
In our running example,
\hl{the species }\lstinline|E| and \lstinline|ES|\hl{, mutually activated through $R_1$
and $R_{2}$, form a conservation law,}
which violates the necessary conditions for multistationarity.
The three Corollaries~\ref{cor:once},~\ref{cor:reversible} and~\ref{cor:pinv}
thus rule out all the cases for which Thomas's condition was
satisfied in this example.

Furthermore, in~\cite{Soule03complexus} and later in~\cite{Soliman13bmb}, the two following
corollaries relying on some particular graph rewirings were also mentionned without much clue to check them:

\begin{corollary}\label{cor:inversed}

   A necessary condition for the multistationarity of a biochemical reaction
   system is that there exist positive cycles fulfilling condition 2 of
   Theorem~\ref{thm:fullrank} in the influence \hl{multigraph} corresponding to its
   Jacobian, \textbf{and} in any graph obtained from it choosing a set of
   species and by reversing the sign of all arcs that have as target some
   species belonging to that set.

\end{corollary}

\begin{corollary}\label{cor:permuted}

   A necessary condition for the multistationarity of a biochemical reaction
   system is that there exist positive cycles fulfilling condition 2 of
   Theorem~\ref{thm:fullrank} in the influence \hl{multigraph} corresponding to its
   Jacobian, \textbf{and} in any graph obtained from it by choosing a permutation
   of the species and by rewiring the arcs' target according to the
   permutation.

\end{corollary}

\section{Graph-Theoretic Algorithm for Proving Non-Multistationarity}\label{sec:alg}

\subsection{{Computing the Labelled Influence Multigraph of a Reaction Model written in SBML}}

The signs of the arcs in the reaction-labelled influence
multigraph of a reaction system, are given by the sign of $\partial v_i/\partial x_j$ instead of that
of $\partial f_i/\partial x_j$. Even without precise kinetic values,
this can be easily computed under the general condition of well-formedness of the reactions~\cite{FMRS18tcbb,FS08tcs}.
This condition is satisfied by the commonly used kinetics such as
mass action law, Michaelis-Menten and
Hill kinetics, \hl{and provides a sanity check for the writing in SBML of ODE models~}\cite{FGS15tcs}.

In the following, and to ensure in a simple and systematic way
that the structure of the reactions, and of the computed influence \hl{multigraph}, do correspond to the continuous dynamics of the model, 
all SBML models considered here are first automatically \emph{sanitized} as explained
in~\cite{FGS15tcs}, by exporting the system of ordinary
differential equations, and reimporting it as a well-formed reaction system.
Algorithm~\ref{alg:sanitize} summarizes the main steps of this procedure.

\begin{algorithm}[htb]
\begin{algorithmic}
   \Function{extract\_labelled\_influence\_graph}{$sbmlModel$}
   \State $Model\gets$\Call{load\_sbml\_model}{$sbmlModel$}
   \State $System\gets$\Call{compute\_odes}{$Model$}
   \State $Model\gets$\Call{infer\_reaction\_model\_from\_odes}{$System$}\\
   \Comment{as explained in~\cite{FGS15tcs}}
   \State $Graph\gets$\Call{infer\_influence\_graph}{$Model$}\\
   \Comment{as explained in~\cite{FS08fmsb} but adding reactions as labels on the edges}
   \State \Return $Graph$
   \EndFunction
\end{algorithmic}
\caption{Algorithm for computing the labelled influence multigraph of a reaction model \hl{written} in SBML~\cite{FGS15tcs}.}\label{alg:sanitize}
\end{algorithm}

\hl{This algorithm needs to determine the sign of a partial derivative.
In our implementation this is done by a simple symbolic derivation algorithm and a heuristic to determine the sign of the expressions.
In case of indeterminacy, both signs are assumed.
In general, the result that is computed is thus an over-approximation of the real influence multigraph.}

\subsection{Absence of Positive Circuit with the Conditions of Cor.~\ref{cor:once}~\ref{cor:reversible}~\ref{cor:pinv}}

Tarjan's \hl{depth-first tree traversal of a graph} provides a classical algorithm for testing the existence of a circuit,
\hl{by just checking the existence of a back edge during this traversal~}\cite{Tarjan72siam}.
\hl{Generalizing this algorithm to check the absence of circuits 
satisfying the previous conditions on the signs and on the reactions at the origin of the arcs
is however non obvious. This may explain why the previous refined graphical requirements had not been implemented before.

In this section, we present an algorithm which proceeds by graph rewriting.
This algorithm will generalize the following graph simplification rules which show that a graph is acyclic
if and only if it reduces to the empty graph by using them}~\cite{ll88joa,pqr99jco,FL06cor}:
\begin{itemize}

   \item \textbf{IN0($v$)}: Remove vertex $v$ and all associated edges if $v$
      has no incoming edge.

   \item \textbf{OUT0($v$)}: Remove vertex $v$ and all associated edges if $v$
      has no outgoing edge.

   \item \textbf{IN1($v$)}: Remove vertex $v$ if $v$ has exactly one incoming
      edge and connect this edge to all the outgoing edges of $v$.

   \item \textbf{OUT1($v$)}: Remove vertex $v$ if $v$ has exactly one outgoing
      edge and connect all incoming edges to it.

\end{itemize}

In order to check the conditions of Cor.~\ref{cor:once}~\ref{cor:reversible}~\ref{cor:pinv}, 
we consider here labelled \hl{multigraph}s, where each arc is labelled by a couple: its sign
and the reaction from which it originates.
Instead of stopping when the previous rules do not apply, and conclude to the cyclicity of the graph if it is not empty,
we extend this set of rules to any number of incoming or outgoing edges and add a restriction on the created edges that must satisfy the conditions of the corollaries. 
\hl{This is described by the following single graph rewriting rule which subsumes} the four previous ones:

\begin{itemize}
   \item \textbf{INOUTi($v$)}: Remove vertex $v$ if $v$ has exactly $i$ incoming or outgoing
      edges, and create the incoming-outcoming edges labeled by the product of the signs and the union of the reactions
if, and only if, those labels satisfy the conditions of the corollaries.
\end{itemize}
This \hl{generic rewriting rule removes} one vertex and all the attached edges
and creates a new edge for every pair of \textit{incoming-outgoing} edge of
the vertex satisfying the conditions. When creating such arcs, the reactions
and species involved in the process \hl{are memorized in order} to check the conditions given by
the previous corollaries and to eliminate the edges, now representing paths,
that do not respect them.
\hl{In this way, this rewriting rule} preserves
all circuits satisfying the conditions of the corollaries.

This rule is applied successively to the nodes of the graph \hl{by choosing a vertex of minimum degree $i$ at each step.
This is done with a simple data structure that maintains the degree of each vertex.}
This algorithm terminates when the first positive self-loop is found, denoting
that a positive circuit satisfying the conditions of the three
Corollaries~\ref{cor:once},~\ref{cor:reversible} and~\ref{cor:pinv}
has been found in the original graph, or when the graph is empty, proving that
no such circuit exists.
The main steps of the decision procedure are summarized in Alg.~\ref{alg:multistat} and~\ref{alg:red}.

\begin{algorithm}[htb]
\begin{algorithmic}[5]
   \Function{CheckAcyclicity}{$G$}
      \While{\Call{CountVertices}{$G$} $> 0$}
         \State $v\gets$ vertex with the least number of incoming or outgoing arcs
         \State \Call{RemoveVertex}{$v$}
         \For{$L \in \mathit{SelfLoops}$}\label{algline:forloop}
            \If{$L$ is positive}
               \State \Return False
            \Else
               \State Delete $L$
            \EndIf
         \EndFor
      \EndWhile
      \State \Return True
   \EndFunction
\end{algorithmic}
\caption{Acyclicity check.}\label{alg:multistat}
\end{algorithm}

\begin{algorithm}[htb]
\begin{algorithmic}
   \Procedure{RemoveVertex}{$v$}
      \For{$w\lra_l v \in IncomingEdges(v)$}
         \For{$v\lra_m x \in OutgoingEdges(v)$}
            \State Create a new label $n\gets l\cdot v\cdot m$
            \If{$n$ does not contain twice the same reaction, a reaction and
            its inverse or all species of a conservation}
            \State Create $w\lra_n x$
            \Else
               \State Discard $n$
            \EndIf
            \State Delete $v\lra_m x$
         \EndFor
         \State Delete $w\lra_l v$
      \EndFor
      \State Delete $v$
   \EndProcedure
\end{algorithmic}
\caption{Graph reduction preserving acyclicity.}\label{alg:red}
\end{algorithm}

\begin{proposition}\label{prop:complexity}

The time complexity of Alg.~\ref{alg:multistat} is $\mathcal{O}\left(k^{2^n}\right)$ where $n$ the number of nodes
and $k$ is the maximum degree of the graph.

\end{proposition}

\begin{proof}

Let us write $k_i$ the maximum indegree or outdegree of the graph after the $i^{\text{th}}$ loop of Alg.~\ref{alg:multistat} ($k_0 = k$). The call to remove a vertex (Alg~\ref{alg:red}) is done in at most $k_i^2$ steps and creates at most the same number of edges, we then have the relation: $k_{i+1} = k_i^2$ which gives: $k_i = k^{2^i}$.

Alg.~\ref{alg:multistat} goes through at most $n$ loops, therefore, the number of steps to complete the algorithm is at most given by:

\[C(k,n) = \sum_{i=1}^{n}k^{2^i} \leqslant \sum_{j=1}^{2^{n-1}} k^{2j} = \mathcal{O} \left(k^{2^n}\right)\]

\end{proof}

\hl{We do not know whether this doubly exponential complexity can be reached in some networks.
It is worth noting that this bound does not take into account the fact that the number of edges strictly decreases
when the rule INOUTi is applied to vertices of degree $i\le 1$},
nor that the edges that do not satisfy the conditions of the corollaries are not created.
Furthermore the degree of the nodes in the initial graph is also a limiting factor as it is generally low \hl{in the context of biochemical networks~}\cite{NMFS16constraints}.
\hl{These considerations explain the much better practical complexity reported in Sections~}\ref{sec:biomodels}, and~\ref{futile}\hl{,
where we will show for instance that the time taken to analyze one model of BioModels is
empirically  $\mathcal{O}(e\log(n))$ where $e$ is the number of edges.}

\subsection{Sign Changes}

\hl{We show here that} the condition given by Corollary~\ref{cor:inversed} can be done by solving a linear system in 
\hl{Galois field} $GF(2)$, \hl{i.e.~$\mathbb Z/2\mathbb Z$}, in which each species is a variable (valued to 1 if the sign of the incoming arcs needs to be reversed). Each simple loop satisfying Corollaries~\ref{cor:once},~\ref{cor:reversible} and~\ref{cor:pinv} in the graph is then modelled by an equation on the sum of all the species involved in the loop, equal to 0 if the loop is negative and 1 otherwise.

As an example, let us consider the influence graph shown in Fig.~\ref{fig:mapk}. This graph contains two positive circuits which satisfy the three Corolaries~\ref{cor:once},~\ref{cor:reversible} and~\ref{cor:pinv} ($ K \xrightarrow[]{R_1} MK \xrightarrow[]{R_2} K $ and $ K \xrightarrow[]{R_3} MpK \xrightarrow[]{R_4} K $) and only one negative circuit satisfying the same Corollaries ($ K \xrightarrow[]{R_1} MK \xrightarrow[]{R_2} Mp \xrightarrow[]{R_3} K $). The system associated to these three loops is therefore:
\[
\left\lbrace\begin{array}{cccccccc}
             x_K &+& x_{MK} & &         & &        & = 1 \\
             x_K & &        &+& x_{MpK} & &        & = 1 \\
             x_K &+& x_{MK} & &         &+& x_{Mp} & = 0 
        \end{array}\right.
     \]
If this system has a solution, then the reaction graph for which we reverse the sign of every arc that has as target any species which variable evaluates to 1 in the solution, does not contain any positive loop satisfying the three Corollaries~\ref{cor:once},~\ref{cor:reversible} and~\ref{cor:pinv}.

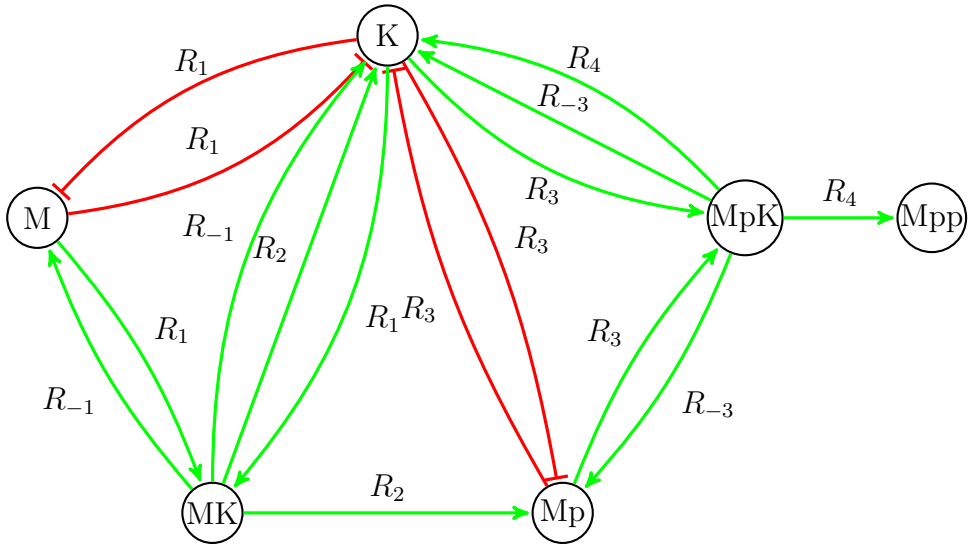
\begin{figure}[htb]
   \begin{center}
      \begin{tikzpicture}
         \matrix[row sep=1.5cm,column sep=1.5cm] {%
			&  & \node[place] (K) {K}; &  &  &  \\
			\node[place] (M) {M}; &  &  &  & \node[place] (MpK) {MpK}; & \node[place] (Mpp) {Mpp}; \\
			\\
			& \node[place] (MK) {MK}; &  & \node[place] (Mp) {Mp}; &  &  \\
         };
        
        \draw[bend right=20,negreg] (K) to node[auto,black,above] {$R_1$} (M);
        \draw[bend right=20,negreg] (M) to node[auto,black] {$R_1$} (K);
        
        \draw[bend left=10,negreg] (Mp) to node[auto,black] {$R_3$} (K);
        \draw[bend left=10,negreg] (K) to node[auto,black] {$R_3$} (Mp);
        
        \draw[bend left=10,posreg] (M) to node[auto,black] {$R_1$} (MK);
        \draw[bend left=10,posreg] (MK) to node[auto,black] {$R_{-1}$} (M);
        
        \draw[posreg] (MK) to node[auto,black] {$R_{2}$} (K);
        \draw[bend left=20,posreg] (MK) to node[auto,black] {$R_{-1}$} (K);
        \draw[bend left=20,posreg] (K) to node[auto,black] {$R_{1}$} (MK);
        
        \draw[posreg] (MK) to node[auto,black] {$R_{2}$} (Mp);
        
        \draw[bend left=10,posreg] (Mp) to node[auto,black] {$R_{3}$} (MpK);
        \draw[bend left=10,posreg] (MpK) to node[auto,black] {$R_{-3}$} (Mp);
        
        \draw[posreg] (MpK) to node[auto,black,above] {$R_{-3}$} (K);
        \draw[bend right=20,posreg] (MpK) to node[auto,black,above] {$R_{4}$} (K);
        \draw[bend right=20,posreg] (K) to node[auto,black,below] {$R_{3}$} (MpK);
        
        \draw[posreg] (MpK) to node[auto,black] {$R_{4}$} (Mpp);
      \end{tikzpicture}
   \end{center}
   \caption{influence multigraph of $M + K \rightleftarrows MK \lra K + Mp
   \rightleftarrows MpK \lra K + Mpp$. Negative self-loops are omitted for clarity}\label{fig:mapk}
\end{figure}

Solving such system can be done by using a simple Gaussian elimination. This process is applied every time a new loop is found by adding the corresponding equation in the system and checking the new equation does not yield a contradiction (which can only be the equation $0 = 1$). This process allows us not to compute every possible loop in the graph if a contradiction emerges.

The previous system obtained from Fig.~\ref{fig:mapk} has two solutions, one of which is $x_K = x_{Mp} = 1$ and the over variables are put to zero. Therefore, the influence graph for which the sign of the incoming arcs for nodes $K$ and $Mp$ are reversed (Fig.~\ref{fig:mapk_sign}) does not contain any positive circuit satisfying Corollaries~\ref{cor:once},~\ref{cor:reversible} and~\ref{cor:pinv}. The biochemical system cannot display any multistationarity.

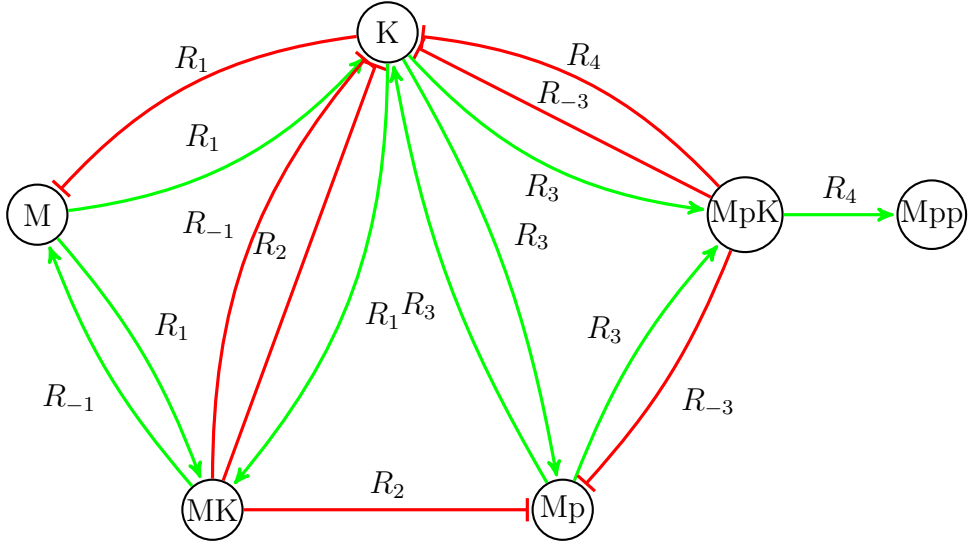
\begin{figure}[htb]
   \begin{center}
      \begin{tikzpicture}
         \matrix[row sep=1.5cm,column sep=1.5cm] {%
			&  & \node[place] (K) {K}; &  &  &  \\
			\node[place] (M) {M}; &  &  &  & \node[place] (MpK) {MpK}; & \node[place] (Mpp) {Mpp}; \\
			\\
			& \node[place] (MK) {MK}; &  & \node[place] (Mp) {Mp}; &  &  \\
         };
        
        \draw[bend right=20,negreg] (K) to node[auto,black,above] {$R_1$} (M);
        \draw[bend right=20,posreg] (M) to node[auto,black] {$R_1$} (K);
        
        \draw[bend left=10,posreg] (Mp) to node[auto,black] {$R_3$} (K);
        \draw[bend left=10,posreg] (K) to node[auto,black] {$R_3$} (Mp);
        
        \draw[bend left=10,posreg] (M) to node[auto,black] {$R_1$} (MK);
        \draw[bend left=10,posreg] (MK) to node[auto,black] {$R_{-1}$} (M);
        
        \draw[negreg] (MK) to node[auto,black] {$R_{2}$} (K);
        \draw[bend left=20,negreg] (MK) to node[auto,black] {$R_{-1}$} (K);
        \draw[bend left=20,posreg] (K) to node[auto,black] {$R_{1}$} (MK);
        
        \draw[negreg] (MK) to node[auto,black] {$R_{2}$} (Mp);
        
        \draw[bend left=10,posreg] (Mp) to node[auto,black] {$R_{3}$} (MpK);
        \draw[bend left=10,negreg] (MpK) to node[auto,black] {$R_{-3}$} (Mp);
        
        \draw[negreg] (MpK) to node[auto,black,above] {$R_{-3}$} (K);
        \draw[bend right=20,negreg] (MpK) to node[auto,black,above] {$R_{4}$} (K);
        \draw[bend right=20,posreg] (K) to node[auto,black,below] {$R_{3}$} (MpK);
        
        \draw[posreg] (MpK) to node[auto,black] {$R_{4}$} (Mpp);
      \end{tikzpicture}
   \end{center}
   \caption{influence multigraph of $M + K \rightleftarrows MK \lra K + Mp
   \rightleftarrows MpK \lra K + Mpp$ for which arcs ending in $K$ and $Mp$ have changed sign. Self-loops are omitted for clarity}\label{fig:mapk_sign}
\end{figure}

Gaussian elimination can be done directly on the species involved in the loop by noticing that adding two equations of the system in $GF(2)$ corresponds to taking the symmetric difference between the list of species of the two loops involved and changing the sign accordingly. This last process is described in Alg.~\ref{alg:gauss}. This function is to be
incorporated in Alg.~\ref{alg:multistat} by replacing the body of the
\textbf{for} loop starting on line~\ref{algline:forloop}, stopping on False
and continuing on True.

\begin{algorithm}[htb]
\begin{algorithmic}
   \Function{AddToLoopSystem}{$l$}
      \For{$m \in LoopSystem$}
         \State Let $P \gets Pivot(m)$
         \If{$P \in l$}
            \State $species(l) \gets species(l)\ \Delta\ species(m)$
            \State $Sign(l) \gets Sign(l) + Sign(m)$
         \EndIf
      \EndFor
      \If{$l$ is positive and $Species(l) = \emptyset$}
         \State \Return False
      \Else
         \State Add $l$ to $LoopSystem$
         \State $Pivot(l) \gets \mathit{FirstSpecies}(l)$
         \State \Return True
      \EndIf
   \EndFunction
\end{algorithmic}
\caption{Gaussian Elimination}\label{alg:gauss}
\end{algorithm}

\subsection{Permutations}

Checking the condition given by Corollary~\ref{cor:permuted} requires more computation as we do not know beforehand the effects of applying a permutation on the circuits of the graph.

As an example, if we apply a swapping between $K$ and $MpK$ in the previous case of Fig.~\ref{fig:mapk}, i.e.~changing the target of every edge that points to $K$ to $MpK$ (including self-loops) and vice-versa, we obtain the graph shown in Fig~\ref{fig:mapk_swap} for which there are 4 positive loops ($K \xrightarrow[]{R_1} MpK \xrightarrow[]{R_4} K$, $K \xrightarrow[]{R_1} MpK \xrightarrow[]{R_{-3}} K$, $K \xrightarrow[]{R_3} MpK \xrightarrow[]{R_4} K$ and $K \xrightarrow[]{R_1} MK \xrightarrow[]{R_2} Mp \xrightarrow[]{R_3} K$).
Changing the sign of $K$ would then transform those positive circuits into negative ones and once again rule out the possibility for multistationarity.
Note that because of the conditions of the original theorems
of~\cite{Soule03complexus} on the diagonal of the Jacobian, only the vertices that
have at least one incoming and one outgoing arc are considered for rewiring.
In other words \textbf{IN0} and \textbf{OUT0} are performed before
permutations.

\hl{Because we did not find any efficient algorithm to propagate target permutation constraints, 
we restricted ourselves to simple permutations made of \emph{one single swapping} between tow molecular species not eliminated by \textbf{IN0} and \textbf{OUT0},
and systematically tried beforehand in a generate-and-test manner.}

\begin{figure}[htb]
   \begin{center}
      \begin{tikzpicture}
         
         \matrix[row sep=2.5cm,column sep=2.5cm] {%
         \node[place] (M) {M};&&\node[place] (K) {K};&\\
         &\node[place] (Mp) {Mp};&&\\
         \node[place] (MK) {MK};&&\node[place] (MpK) {MpK};&\node[place] (Mpp) {Mpp};\\
         };
         
         \coordinate[shift={(-10mm,-2mm)}] (n) at (MK.south east);

		\draw[bend right=20,negreg] (K) to node[auto,black, below] {$R_1$} (MpK);
		\draw[negreg] (K) to node[auto,black, below] {$R_3$} (MpK);   
		\draw[bend right = 20, negreg] (MpK) to node[auto,black, above] {$R_4$} (K);
		\draw[bend right = 40, negreg] (MpK) to node[auto,black, above] {$R_{-3}$} (K);

        \draw[bend right=20,negreg] (K) to node[auto,black,above] {$R_1$} (M);
        \draw[bend right=45, negreg] (M) to node[auto,black] {} (n) to node[auto,black] {$R_1$} (MpK);
        
        \draw[bend left=10,negreg] (Mp) to node[auto,black] {$R_3$} (MpK);
        \draw[bend left=10,negreg] (K) to node[auto,black] {$R_3$} (Mp);
        
        \draw[bend left=10,posreg] (M) to node[auto,black] {$R_1$} (MK);
        \draw[bend left=10,posreg] (MK) to node[auto,black] {$R_{-1}$} (M);
        
        \draw[posreg] (MK) to node[auto,black, below] {$R_{2}$} (MpK);
        \draw[bend left=20,posreg] (MK) to node[auto,black] {$R_{-1}$} (MpK);
        \draw[bend right=30,posreg] (K) to node[auto,black, above] {$R_{1}$} (MK);
        
        \draw[posreg] (MK) to node[auto,black] {$R_{2}$} (Mp);
        
        \draw[bend left=10,posreg] (Mp) to node[auto,black] {$R_{3}$} (K);
        \draw[bend left=10,posreg] (MpK) to node[auto,black] {$R_{-3}$} (Mp);
        
        \draw[posreg] (MpK) to node[auto,black] {$R_{4}$} (Mpp);
      \end{tikzpicture}
   \end{center}
   \caption{influence multigraph of $M + K \rightleftarrows MK \lra K + Mp
   \rightleftarrows MpK \lra K + Mpp$ for which arcs ending in $K$ and $MpK$
have been swapped. Self-loops are omitted for clarity}\label{fig:mapk_swap}
\end{figure}
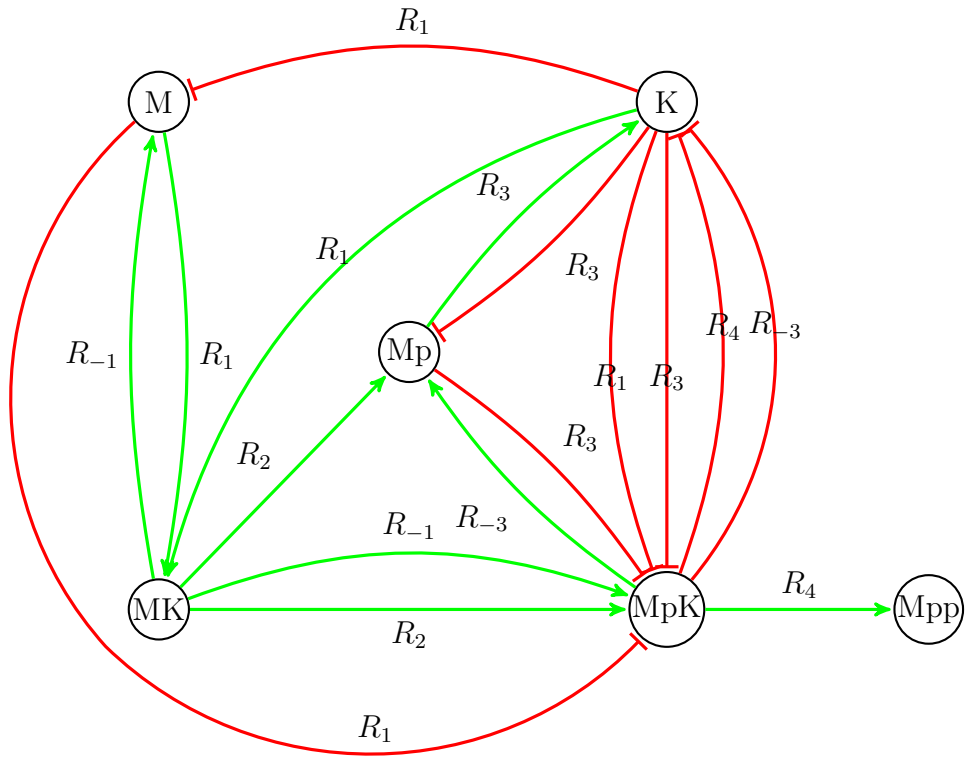

\section{Evaluation on the BioModels Repository}\label{sec:biomodels}

\subsection{Reaction Networks from BioModels}

To evaluate the information brought by \hl{the previous graphical requirements for multistationarity and the performance of our graph rewriting algorithms}, we downloaded the
latest release of the BioModels database\footnote{\url{http://biomodels.net/},
31st release, dated 26th of June, 2017}~\cite{CLN13issb} and applied our
method in a systematic way. First, a labelled influence \hl{multigraph} is extracted as
per Alg.~\ref{alg:sanitize}. Out of the 640 curated models, the extraction led
to 506 models with a non-trivial influence \hl{multigraph}. The other models rely on
events, assignment-rules, etc.\ to enforce their dynamics, or simply do not
contain reaction but flux-balance constraints, gene-regulations, etc.\ even
after re-import of their ODE dynamics~\cite{FGS15tcs}.

\subsection{Results of the Graphical Algorithm}

\begin{table}[htb]
   \begin{center}
      \begin{tabular}{lrrrrr}
         \toprule
         Conditions verified& Number & \multicolumn{2}{r}{Nb of species} & \multicolumn{2}{r}{Computation time}\\
          & of graphs & avg. & max. & avg. (s) & max (s)\\ \midrule
            All graphs & 506 & 21.24 & 430  &  & \\ \midrule
         \hl{No negative circuit} & \hl{70} & \hl{6.87} & \hl{57} & \hl{$< 0.01$} & \hl{$0.05$}\\ \midrule
         No positive circuit & 48 & 3.42 & 18 & $< 0.01$ & $<0.01$\\
         Cor.~\ref{cor:once}~\ref{cor:reversible}~\ref{cor:pinv} & 105 &
            6.22 & 46 & $< 0.01$ & $<0.01$\\
         Cor.~\ref{cor:once}~\ref{cor:reversible}~\ref{cor:pinv}~\ref{cor:inversed} & 160 &
            8.23 & 54 & $< 0.01$ & $0.05$\\
         Cor.~\ref{cor:once}~\ref{cor:reversible}~\ref{cor:pinv}~\ref{cor:inversed}~\ref{cor:permuted} & 180 & 8.38 & 54 & 5.90 & 980.1 \\
         \bottomrule
      \end{tabular}
   \end{center}
   \caption{Analysis of \hl{the 506 sanitized reaction models from the curated branch of
         the BioModels repository. The table reports the proportion of models, and their size, for which 
         the non-existence of multiple steady states is proved using Thomas's positive circuit condition and using
         the refined conditions expressed in the corollaries described above. 
   The computation times are given for the whole set of models. The maximum computation time is obtained for checking the last condition on model number 574.}\label{tab:biomodels}}
\end{table}

Table~\ref{tab:biomodels} summarizes the results of our experiments.
\hl{It is worth noting that the maximum running time of 50ms for checking our main graphical requirements is remarkably low.
   It concerns all models of the benchmark, including the largest model number 235 that contains 430 species
and an influence multigraph of 1875 arcs.}

\hl{Another} observation is that \hl{not only the
number of models for which the absence of multistationarity is proved}
more than doubles when using
Corollaries~\ref{cor:once}~\ref{cor:reversible}~\ref{cor:pinv} on top of
Thomas's simple condition, the models that are added are of \hl{much larger size} than the one dealt with the original conditions.
Indeed with the simple condition, only very small models
with less than 18 species and linear reactions were shown to have no multistationarity, whereas the stronger conditions
allow us to prove the absence of multistationarity in models of size up to 46 species and including non-linear reactions. This is far below the size of the biggest
models of the BioModels repository (for which the existence of multiple steady states is generally unknown) but shows that the supplementary conditions
do change the scope of use of the method. 

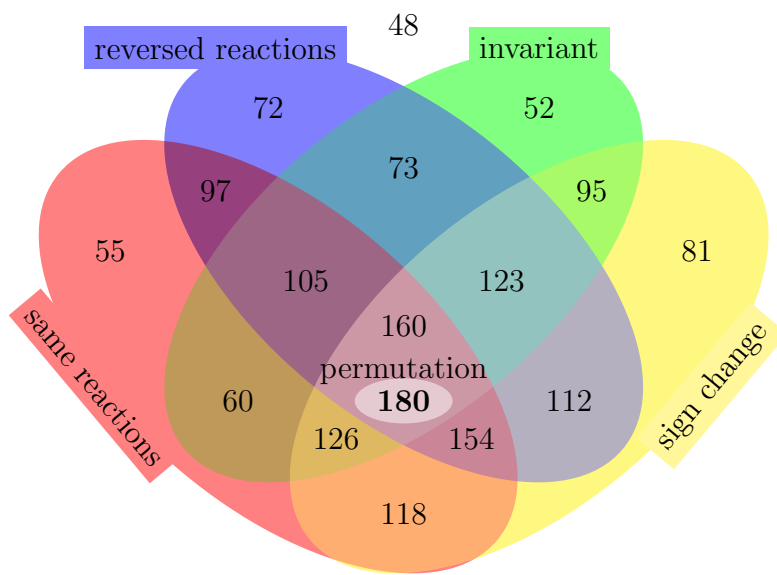
\begin{figure}[htb]
	\begin{center}
		\begin{tikzpicture}[scale=1]
		
			\begin{scope}[transparency group]
			\begin{scope}[blend mode=luminosity]
        
			\fill[opacity=0.5, line width=1pt, rotate=-40, red] (-0.5,-2) ellipse (3.8 and 2);
			\fill[opacity=0.5, line width=1pt, rotate=-40, blue] (0,0) ellipse (3.8 and 2);
			\fill[opacity=0.5, line width=1pt, rotate=40, green] (0,0) ellipse (3.8 and 2);
			\fill[opacity=0.5, line width=1pt, rotate=40, yellow] (0.5,-2) ellipse (3.8 and 2);
			\fill[opacity=0.5, line width=1pt, white] (0,-1.8) ellipse (0.65 and 0.3);
            
			\end{scope}
			\end{scope}
            
			\node (a) at (0,-0.8) {${160}$};
			\node (a) at (0,-1.8) {$\mathbf{180}$};
			\node (a) at (0.9,-2.3) {$154$};
			\node (a) at (-0.9,-2.3) {$126$};
			\node (a) at (0,-3.3) {$118$};
			\node (a) at (2.2,-1.8) {$112$};
			\node (a) at (-2.2,-1.8) {$60$};
			\node (a) at (1.3,-0.2) {$123$};
			\node (a) at (-1.3,-0.2) {$105$};
			\node (a) at (3.9,0.2) {$81$};
			\node (a) at (-3.9,0.2) {$55$};
			\node (a) at (0,1.3) {$73$};
			\node (a) at (2.5,1) {$95$};
			\node (a) at (-2.5,1) {$97$};
			\node (a) at (1.8,2.1) {$52$};
			\node (a) at (-1.8,2.1) {$72$};
			\node (a) at (0,3.2) {$48$};
            
			\node (a) at (0,-1.4) {permutation};
			\node[fill=red!50, rotate=-50] (a) at (-4.1,-1.6) {same reactions};
			\node[fill=blue!50] (a) at (-2.5,2.9) {reversed reactions};
			\node[fill=green!50] (a) at (1.8,2.9) {invariant};
			\node[fill=yellow!50, rotate=50] (a) at (4.1,-1.6) {sign change};
        
		\end{tikzpicture}
        
	\end{center}
	\caption{\hl{Number of models among the 506 curated reaction models of BioModels for which multistationarity can be ruled out by using respectively original Thomas's positive circuit condition, Cor.~}\ref{cor:once} (no same reactions),~\ref{cor:reversible} (no reversed reactions),~\ref{cor:pinv} (no invariant) and~\ref{cor:inversed} (sign change), plus~\ref{cor:permuted} (permutation).}\label{fig:venn}
\end{figure}

\hl{Fig.}~\ref{fig:venn} shows a Venn diagram which details the contribution of the different graphical conditions.
One can note that Cor.~\ref{cor:pinv} was in fact useful in only eight of the new models found with the other two corollaries
combined.
The condition of Cor.~\ref{cor:inversed} (sign change of all incoming edges to a set of vertices) is responsible for concluding 
to the absence of multistationarity in 55 more models, of size up to 54 species.

Corollary~\ref{cor:permuted} \hl{allows us to rule out multistationarity} in 20 new models, but with the same maximum size.
Even with the restriction to \hl{single transpositions} as explained above, the \hl{maximum running time on the whole benchmark} becomes much higher than
for the simpler conditions, by five orders of magnitude. 
However, the increase in the number of models for which
multistationarity is proved not to be possible 
with this restricted strategy is relatively high (20) and thus encouraging \hl{for further improvements.
Indeed, better heuristics or more efficient
propagation of the permutation constraint, might lead to even more conclusive results on even larger size problems}.

\subsection{Practical Complexity}

\begin{figure}[htb]
   \begin{center}
      \includegraphics[width=\textwidth]{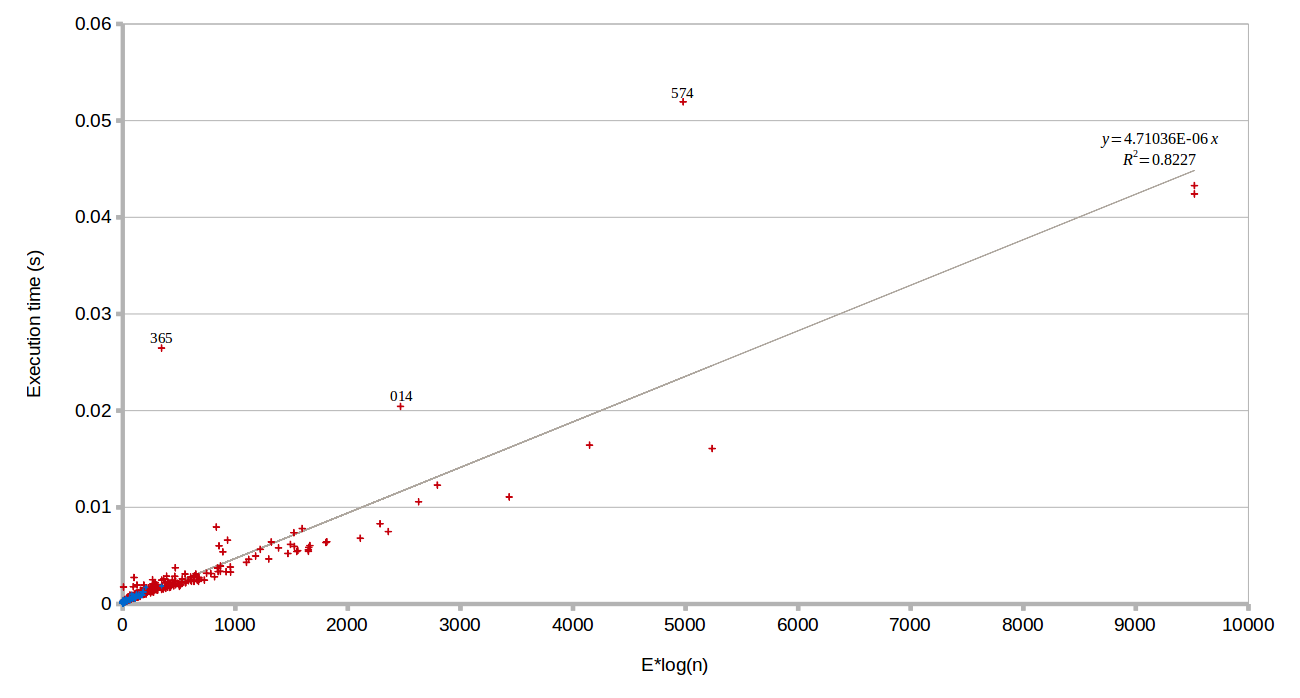}
      \caption{Execution time of the 506 models of BioModels tested with Alg.~\ref{alg:multistat} and~\ref{alg:gauss}
with the conditions of Cor.\ref{cor:once}~\ref{cor:reversible}~\ref{cor:pinv}~\ref{cor:inversed}, relatively to the value of $e\log(n)$. 
The points in blue represent the models for which multistationarity is proved impossible, and the ones in red, those for which the algorithm exhibits a circuit that satisfies the conditions of the Corollaries.}\label{fig:exec_time}
   \end{center}
\end{figure}

The computation times presented in Table~\ref{tab:biomodels} with the use of the first four corollaries are far better than the theoretical complexity bound given in Prop~\ref{prop:complexity}. 
It is known that reducing the graph with only the 4 original rules IN0, IN1, OUT0 and OUT1 can be done with a time complexity in $\mathcal{O}\left(e\log(n)\right)$~\cite{ll88joa} where $e$ is the number of edges in the graph. 
In Alg.~\ref{alg:multistat}, these simplification rules are \hl{in fact} used with a higher priority than the rule INOUTi which are taken with the increasing order on $i$. 
Fig.~\ref{fig:exec_time} plots the computation time for each model relatively to the value of $e\log(n)$.
The linear shape of the curve suggests that the empirical time complexity on BioModels is close to $\mathcal{O}\left(e\log(n)\right)$,
i.e.~the complexity of the 4 original rules, and that the extra rules INOUTi with $i \ge 2$
(although used to conclude on non multistationarity in 65 \hl{over the 160} models)
do not significantly increase the computation time apart from very few cases (models 014, 365 and 574) up to 50ms.

\hl{More precisely}, the rule INOUTi is used with an average maximum value of $i = 3.5$ and 139 models do not use more than $i = 1$ which explains why it does not add much to the computation time. Models 014 and 365 use the rule INOUTi with $i=36$ and $i=98$ respectively which explains their higher computation time. Model 574 uses the rule INOUTi with $i = 8$ which is not uncommon but this graph is dense and each node has at least a degree of 3. In this case, the simple rules with $i = 0$ or $i = 1$ that give a complexity in $\mathcal{O}\left(e\log(n)\right)$ are never used which \hl{basically} explains why this model is the longest to check.

\subsection{Comparison to the Jacobian-based Symbolic Computation Method}

In~\cite{FW13bi}, \hl{Feliu and Wiuf have presented a symbolic computation algorithm implemented in Maple 16
to directly check the existence of roots of some matrix determinant which is equivalent to a non-injectivity property implied by the existence of multiple steady states.
That condition is in principle stronger than the graphical requirements we consider.
Interestingly, they evaluated their algorithm on BioModels, with a version at that time of 365 curated models.
Their method showed that 31.6\% of the networks do not have multiple steady states,
i.e.~the same proportion as us (160 out of 506 networks) when checking the first four corollaries,
but their Maple program failed by memory overflow on 8\% of the networks
whereas our maximum computation time is 0.05s.}

\hl{
Furthermore, the proportion of conclusive analyses raises in our case to 35,5\% (180 out of 506) by using the last corollary
but currently with a high computational cost and a restricted implementation of that condition.
}

\section{Analysis of Multiple Phosphorylation Cycles and MAPK Signalling Models}

\subsection{Wang and Sontag's Futile Phosphorylation Cycles}\label{futile}

\begin{figure}
   \begin{center}
      \includegraphics[width=.5\textwidth]{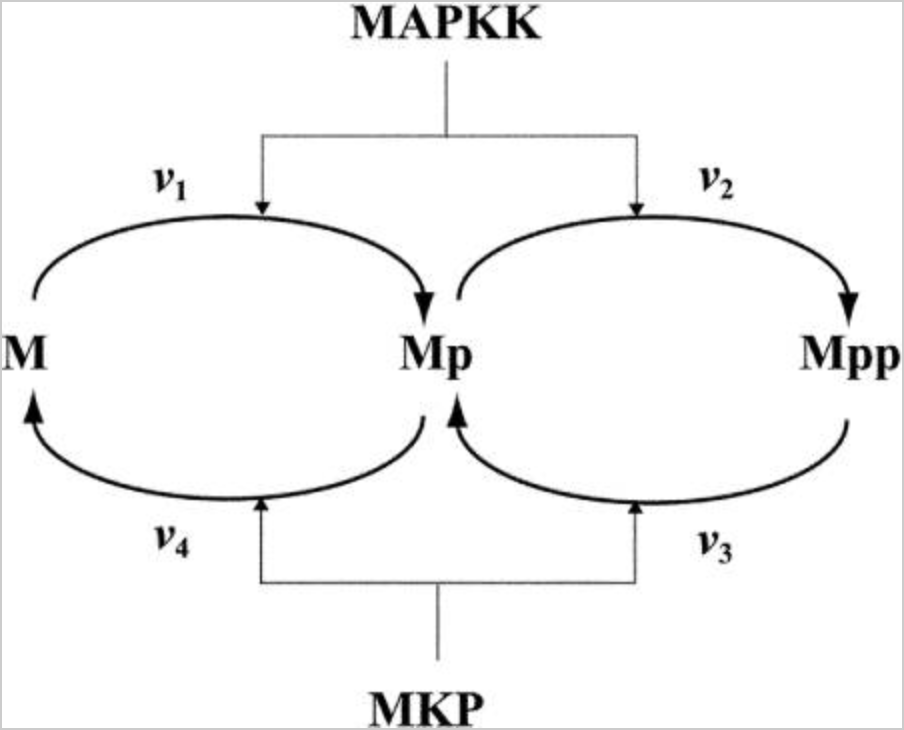}
      \caption{Fig.~1 of~\cite{MHK04jcb} displaying the general
      double-phosphorylation cycle involved in all the models studied by
   Markevich et al.~and by Wang and Sontag~\cite{WS08jmb}.}\label{fig:markevich}
   \end{center}
\end{figure}

In~\cite{FW13bi}, \hl{Feliu and Wiuf also evaluated their method on the
the $r$-site phosphorylation cycles of Wang and Sontag
who showed the existence of multiple steady states in those networks for $r\ge 2$}~\cite{WS08jmb}\hl{.
The case $r=2$, schematized in Fig.}~\ref{fig:markevich}\hl{, was extensively studied by Markevich et al.~in~}\cite{MHK04jcb}\hl{
in a series of models, numbered from 26 to 31 in the
BioModels repository (see~}\cite{GSF10bi} \hl{for the model reduction relationships between these models found by subgraph epimorphisms), 
showing in all cases the existence of multiple steady states.}

\hl{The symbolic computation method used by Feliu and Wiuf 
grew rapidly in time as a function of the number $r$ of phosphorylations, and became impractical after $r=17$ for which it needed 1200 seconds.
Our graphical method has a very low computational complexity on these networks, taking only 1.2s for $r=1000$.
It is worth noting that our method checks necessary conditions for the non-injectivity of the system whereas the symbolic method of Feliu and Wiuf directly determines that property. 
In both cases, though, one cannot conclude that the system does have multiple stationary states since the non-injectivity property is itself a necessary not sufficient condition.
}

\hl{Table}~\ref{tab:r_site_phosphorylation} \hl{summarizes our results.
The first two computation time columns refer to the original model of Wang and Sontag
in which each phosphorylation and dephosphorylation transformation is modelled
by three reactions with mass action law kinetics
with an explicit representation
of the intermediary complexes, 
by repeating the following pattern:}
\begin{equation}
\left\{\begin{array}{l}

       E+S_i \rightleftarrows ES_i \rightarrow E+S_{i+1} \\
       F+S_{i+1} \rightleftarrows FS_{i+1} \rightarrow F+S_i

\end{array}\right.
\label{eq:WangSontag}
\end{equation}
\hl{This gives the following differential equations:}
$$
\begin{array}{rclr}
\dfrac{dS_0}{dt} &=& -k_{\on_0}S_0E + k_{\off_0}ES_0 + l_{cat_0}FS_1 \\
\dfrac{dS_i}{dt} &=& -k_{\on_i}S_iE + k_{\off_i}ES_i + k_{cat_{i-1}}ES_{i-1}  \\ && - l_{\on_{i-1}}S_if +  l_{\off_{i-1}}FS_i + l_{cat_i}FS_{i+1} &, i = 1, ..., n-1\\
\dfrac{dES_j}{dt} &=& k_{\on_j}S_jE - \left(k_{\off_j} + k_{cat_j}\right)ES_j &, j = 0, ..., n-1\\\\
\dfrac{dFS_k}{dt} &=& l_{\on_{k-1}}S_kf - \left(l_{\off_{k-1}} + l_{cat_{k-1}}\right)FS_k &, k = 1, ..., n\\
\end{array}
$$
\hl{The computation times given in Table}~\ref{tab:r_site_phosphorylation} \hl{indicate that, on these networks, our method 
has an empirical complexity of the order of $10^{-3}r^2$.}

\begin{table}[!htb]
   \begin{center}
      \begin{tabular}{ccccc} 
         \toprule
         \hl{$r$} & \hl{Jacobian Method}~\cite{FW13bi} & \multicolumn{3}{r}{\hl{Graphical Method (Alg}~\ref{alg:multistat} \hl{\&}~\ref{alg:gauss}\hl{)}}\\
           &  \hl{model }\eqref{eq:WangSontag}& \eqref{eq:WangSontag} &
         \eqref{eq:Markevich}& \eqref{eq:CatalyticReactions} \\ \midrule
         \hl{ 1    }&\hl{ 4        }&\hl{ 0.3 }&\hl{ 1.6 }&\hl{ 0.6 }\\
         \hl{ 2    }&\hl{ 75       }&\hl{ 0.5 }&\hl{ 2 }&\hl{ 1 }\\
         \hl{ 3    }&\hl{ 44       }&\hl{ 1 }&\hl{ 2 }&\hl{ 1 }\\
         \hl{ 4    }&\hl{ 81       }&\hl{ 1 }&\hl{ 3 }&\hl{ 1 }\\
         \hl{ 5    }&\hl{ 191      }&\hl{ 1 }&\hl{ 4 }&\hl{ 1 }\\
         \hl{ 6    }&\hl{ 256      }&\hl{ 1 }&\hl{ 4 }&\hl{ 1 }\\
         \hl{ 7    }&\hl{ 444      }&\hl{ 1 }&\hl{ 5 }&\hl{ 1 }\\
         \hl{ 8    }&\hl{ 795      }&\hl{ 2 }&\hl{ 5 }&\hl{ 1 }\\
         \hl{ 9    }&\hl{ 1169     }&\hl{ 2 }&\hl{ 6 }&\hl{ 2 }\\
         \hl{ 10   }&\hl{ 2195     }&\hl{ 2 }&\hl{ 6 }&\hl{ 2 }\\
         \hl{ 11   }&\hl{ 3998     }&\hl{ 2 }&\hl{ 6 }&\hl{ 2 }\\
         \hl{ 12   }&\hl{ 7696     }&\hl{ 2 }&\hl{ 7 }&\hl{ 2 }\\
         \hl{ 13   }&\hl{ 15180    }&\hl{ 2 }&\hl{ 7 }&\hl{ 2 }\\
         \hl{ 14   }&\hl{ 32180    }&\hl{ 3 }&\hl{ 7 }&\hl{ 2 }\\
         \hl{ 15   }&\hl{ 67740    }&\hl{ 3 }&\hl{ 7 }&\hl{ 2 }\\
         \hl{ 16   }&\hl{ 171700   }&\hl{ 3 }&\hl{ 8 }&\hl{ 2 }\\
         \hl{ 17   }&\hl{ 1199000  }&\hl{ 4 }&\hl{ 8 }&\hl{ 2 }\\
         \hl{ 50   }&\hl{ $\times$ }&\hl{ 12 }&\hl{ 17 }&\hl{ 4 }\\
         \hl{ 100   }&\hl{ $\times$ }&\hl{ 26 }&\hl{ 40 }&\hl{ 6 }\\
         \hl{ 500   }&\hl{ $\times$ }&\hl{ 343 }&\hl{549 }&\hl{ 34}\\
         \hl{ 1000 }&\hl{ $\times$ }&\hl{ 1200  }&\hl{ 1874  }&\hl{ 98 }\\

         \bottomrule
      \end{tabular}
   \end{center}
   \caption{Execution times given in milliseconds for the analysis of the $r$-site phosphorylation system of~\cite{WS08jmb},
first as reported in~\cite{FW13bi} for the Jacobian method using symbolic computation, then obtained with our graphical algorithm on the same model and on two variants concerning the writing of the dephosphorylation and phosphorylation reactions.}\label{tab:r_site_phosphorylation}
\end{table}

\hl{The third computation time column refers to the writing of the 
dephosphorylations with two intermediate complexes, as follows:}
\begin{equation}
\left\{\begin{array}{l}

       E+S_i \rightleftarrows ES_i \rightarrow E+S_{i+1} \\
       F+S_{i+1} \rightleftarrows FS_{i+1}^\star \rightarrow FS_i \rightleftarrows F+S_i

\end{array}\right.
\label{eq:Markevich}
\end{equation}
\hl{
This writing of the dephosphorylations corresponds to the first model of Markevich et al.~}\cite{MHK04jcb}.
\hl{On this reaction pattern }\eqref{eq:Markevich}\hl{, our graph algorithm
   has execution times similar to those obtained on reaction pattern
}\eqref{eq:WangSontag}.
This is due to the resemblance of their influence multigraphs.
Nevertheless, the rule INOUTi is used \hl{here} with $i \le 12$, 
while on model pattern~\eqref{eq:WangSontag}\hl{ it is used with value at most $9$.
This is responsible for a slight difference in response time.}

\hl{The second model of Markevitch et al.~}\cite{MHK04jcb} \hl{
is a reduction of the previous model using Michaelian kinetics.
The intermediary complexes are eliminated
but}
the writing of the kinetics for the
dephosphorylation of Mp by phosphatase MKP3, named $v_4$ in the original
article~\cite{MHK04jcb}, is \hl{not a naive Michaelis-Menten kinetics but the following one}:
\[v_4 = \frac{k_4^{cat}\cdot{[\mathit{MKP3}]}_{tot}\cdot [Mp]/K_{m3}}{(1 +
[Mpp]/K_{m3} + [Mp]/K_{m4} + [M]/K_{m5})}\]
\hl{Mpp appears as inhibitor in this
kinetic expression to represent} the sequestration of the phosphatase in the
reversible last step of dephosphorylation, 
Such a sequestration results in a negative influence of Mpp on M and
a positive influence of Mpp on Mp (i.e., a positive term in
$\frac{\partial\dot{Mp}}{\partial Mpp}$) as this reaction consumes Mp to
produce M.
Intuitively the fact that Mpp can actively sequestrate the phosphatase MKP3
makes it inhibit the dephosphorylation of Mp and therefore stabilizes Mp.
\hl{Therefore, while the positive circuit between Mp and Mpp by $v_2$ and $v_3$ that can be easily seen from Fig.~}\ref{fig:markevich} 
\hl{is immediately rule out since $v_2$ and $v_3$ are opposite reactions},
the complex Michaelian kinetics of \cite{MHK04jcb}
gives rise to a completely different positive circuit between those two
species (with kinetics $v_4$ and $v_2$). This circuit cannot be removed by sign-changes or
any single swap and is indeed responsible for the appearance of bistability in those models.

\hl{The fourth computation time column refers to that model structure,
with catalytic reactions instead of intermediary complexes, but using mass action law (or simple Michaelis-Menten) kinetics, 
with the following pattern:} 
\begin{equation}
\left\{\begin{array}{l}

       E+S_i \rightarrow E+S_{i+1} \\
       F+S_{i+1} \rightarrow  F+S_i

\end{array}\right.
\label{eq:CatalyticReactions}
\end{equation}
\hl{The differential equations for this pattern with simple Michaelis-Menten kinetics are as follows :}
$$
\begin{array}{rclr}
\dfrac{dS_0}{dt} &=& -\dfrac{V_0ES_0}{K_0+S_0} + \dfrac{V_0^*fS_{1}}{K_0^*+S_1}\\
\dfrac{dS_i}{dt} &=& -\dfrac{V_iES_i}{K_i+S_i} + \dfrac{k_{i-1}ES_{i-1}}{K_{i-1}+S_{i-1}} \\
&& + \dfrac{V_i^*fS_{i+1}}{K_i^*+S_{i+1}} - \dfrac{V_{i-1}^*fS_i}{K_{i-1}^* + S_i} &~~~~, i=1,...,n-1 \\
\dfrac{dS_n}{dt} &=& \dfrac{V_{n-1}ES_{n-1}}{K_{n-1}+S_{n-1}} - \dfrac{V_{n-1}^*fS_n}{K_{n-1}^*+S_n}
\end{array}
$$
\hl{In this modelling of the system, the possibilities of multistationarity disappear,
and this is shown by the result of our graphical algorithm.
On large instances, the computation time is also significantly lower.
This is because the graph does not contain any positive cycle satisfying our conditions,
and the graph algorithm has only to remove nodes using the rule INOUTi with $i\le 1$.
}

\subsection{MAPK Signalling Models}\label{sec:erk}

\begin{figure}[htb]
   \begin{center}
      \includegraphics[width=.8\textwidth]{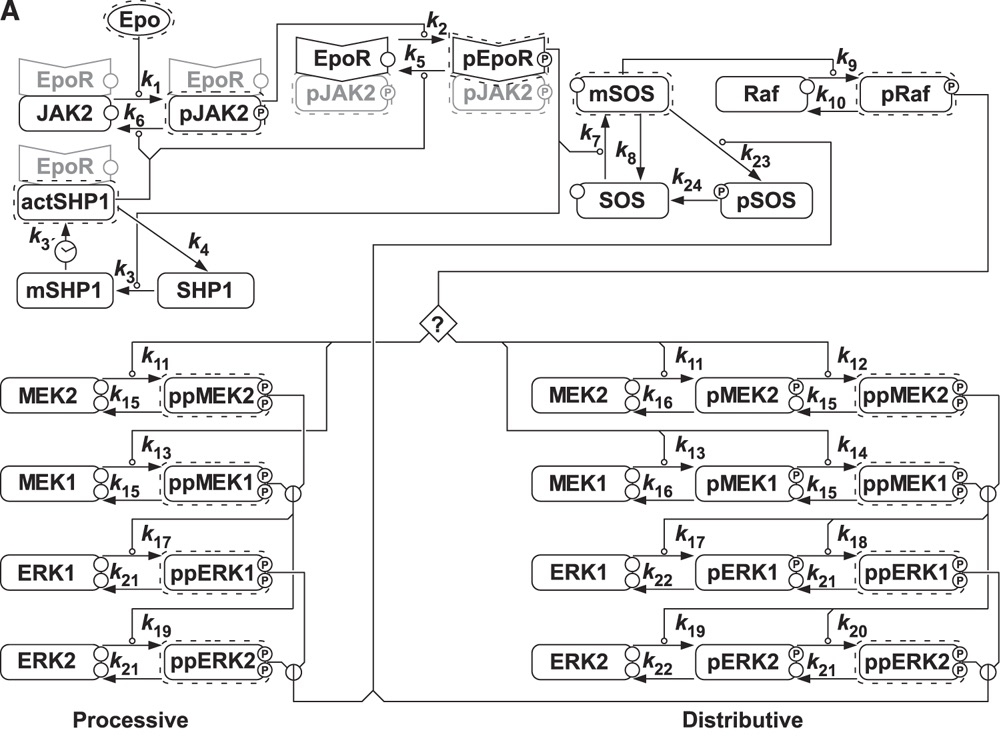}
      \caption{Figure 1A of~\cite{SMHWKKLTK09msb}, on the right is the
      distributive model encoded as model 270 in the BioModels repository.
      Several other variables appear in the model because of the encoding with ODEs of
      Delay Differential Equations involved close to the receptor.}\label{fig:erk}
   \end{center}
\end{figure}

\hl{
Double phosphorylation cycles are integral parts of MAPK signalling cascades
and one might expect to observe multistationarity in MAPK models.
However, as shown in the previous section, this depends on the way the phosphorylation and dephosphorylation reactions are modeled.
}

\hl{Model 270 of the BioModels database} describes a complete Epo-induced ERK signalling cascade, from
receptor binding to cell fate decision, corresponding to the \emph{distributive} model of~\cite{SMHWKKLTK09msb}, schematized in
Fig.~\ref{fig:erk}.
It includes four reversible
double-phosphorylation stages (MEK2, MEK1, ERK1 and ERK2) and many dummy
variables introduced at the beginning of the cascade to encode delay ordinary
equations into simple ODEs. The resulting reaction model has 33 species and 42 reactions,
and leads to a labelled influence \hl{multigraph} containing 126 arcs with many positive and
negative feedback loops. 

During the run of Alg.~\ref{alg:multistat}, only 2
paths are removed thanks to one of the 9 conservation laws (Cor.~\ref{cor:pinv}), but 58 are removed
thanks to Cor.~\ref{cor:once} and~\ref{cor:reversible}, resulting in no single
positive feedback loop satisfying the conditions of our theorem.

The fact that multistationarity is not possible for such a model is consistent
with the data shown by the authors in the article, i.e.~functional dose-response diagrams without hystereses.
This is however not evident, nor perhaps expected, from the
model itself since memory effects resulting in hysteresis might have been
possible at many different places in the network due to the dephosphorylation loops.
\hl{However, as explained in the previous section with model~}\ref{eq:CatalyticReactions} 
\hl{this is not the case when simple Michaelis-Menten kinetics are used}.

These examples show that the existence of multiple steady states in reaction networks
is sensitive to the explicit representation of the intermediate complexes in enzymatic reactions,
or at least to the explicit inclusion of their inhibitors (by sequestration) in the kinetics.
Interestingly, our refined conditions are similarly sensitive to these subtle modelling choices 
and allow us to conclude differently according to the impact of the writing of the reactions
on the multistationarity properties of the system. 
The role of intermediate complexes in multistationarity was analysed in detail in~\cite{FC13rsi}. 
In particular, it was shown that if the network does not have conservation laws, 
then multistationarity cannot arise after the introduction of intermediate complexes.

These remarks also go in the same direction to what has been observed for oscillations in the
MAPK cascade again, where the absence of complexation removes the negative feedbacks
going upwards in the cascade and therefore the negative feedback loops and
the corresponding possibility of oscillations~\cite{VSM08plos}. If the intermediary complexes are
explicitly represented, then oscillations can be found~\cite{QNKS07plos},
without any external negative feedback reaction such as receptor desensitization~\cite{Kholodenko00ejb}.
\hl{In many other networks in} BioModels, Alg.~\ref{alg:multistat}
actually shows as side-effect the existence of numerous negative feedback loops.

\section{Conclusion}\label{sec:conclu}

This experiment is, to our knowledge, the first systematic evaluation of
\hl{graphical requirements for multistationarity in the reaction networks of model repositories in large scale.
We have shown that Thomas's necessary condition for multistationarity, and its refinement for reaction models given in~}\cite{Soliman13bmb},
can be implemented with a graph rewriting algorithm that brings useful information for
many models in BioModels,
by proving the non-existence of multiple steady states independently of the parameter values and of the precise form of the rate functions.
\hl{Though the
original Thomas's conditions show the absence of multistationarity in some small models, 
the refined conditions are conclusive in many more cases: 180 vs~48, including
much bigger models: up to 54 vertices vs~18.
Furthermore this is achieved at a remarkably low computational cost, below 0.05 second per network for the main conditions, 
even on models with several hundreds of molecular species and thousands of influence arcs, currently out of reach of symbolic computation methods}.

It is worth noting that our graph-theoretic algorithm is not limited to reaction systems with mass action law kinetics,
but relies on a simple symbolic derivation algorithm for computing an over-approximation of the signs of the partial derivatives in the Jacobian matrix.
\hl{In case of indeterminacy of the sign, both signs are assumed which may lead to the existence of circuits that would have been ruled out by a more accurate
determination of the sign.}
Our graphical algorithm could also be improved by a more efficient \hl{and complete} use of the \hl{condition dealing with target species permutations},
for instance by recourse to constraint propagation algorithms~\cite{FL06cor} instead of the current generate-and-test procedure for single swappings.
\hl{This might further increase the number of conclusive cases}.
\hl{Another way would be to use the condition noted (*)
in}~\cite{BC09cms} \hl{to rule out
the positive circuits that do not intersect another positive circuit on a
species-to-reaction path, which is necessary for multistationarity.}

Since our procedure, and more precisely Alg.~\ref{alg:sanitize}, goes
through an import of the ODE system \hl{and infers a reaction network, it can be readily used} on dynamical systems that do not stem from reaction
networks but may exhibit similar symmetries. Such use of the refined conditions in general ODE systems would probably benefit much less from the structural
conditions added on top of Thomas's rules, but by identifying similar terms in the
ODEs, our algorithm should be able to automatically prove the absence of
multistationarity in interesting cases, \hl{as also suggested in}~\cite{BC09cms}.

A comparison to Feinberg's CRNT-based approaches would also be interesting,
\hl{%
   by considering the different approaches summarized for instance in Table 3
   of
}~\cite{FW13bi}.
\hl{In particular, our circuit conditions on the influence multigraph depend on the signs of the entries of the Jacobian matrix but are independent 
of not only the values of the kinetic parameters, but also of the form of the reaction rate functions which can be any partially differentiable function,
i.e.~without any monotonicity, non-autocatalytic, or such restriction.}

\hl{Finally, this study focussed on multistationarity, but we saw that most
models of the benchmark also have negative circuits.
A systematic study of the oscillation conditions in reaction model repositories, 
possibly using a similar theoretical refinement of Thomas-Snoussi's necessary
conditions for sustained oscillations}~\cite{Snoussi98jbs},
\hl{would be worth investigating 
as natural systems indeed provide many oscillators and}
even models not conceived to oscillate have been shown capable of exhibiting unexpected sustained oscillations in non-standard conditions~\cite{QNKS07plos}.

\subsection*{Acknowledgments.}
This work has been supported by the bilateral project ANR-17-CE40-0036 
SYMBIONT ({https://www.symbiont-project.org}).
We are also grateful to the reviewers for their useful comments for improving the presentation of our results.









\section*{References}
\bibliographystyle{elsarticle-harv}

\bibliography{contraintes}

\appendix

\subsection*{Afterword \emph{in memoriam} of Ren\'e Thomas} 
\hl{It was under the sun of the University of Marseille, at the CIRM in 2008, that the
second author met Ren\'e Thomas and his wife, in the friendly atmosphere of a
summer school where Ren\'e participated to all talks giving advices with
extreme modesty and continuing providing deep insights during the traditional walk to the Callanques.
In his talk, he mentioned that the gene interactions he had always been considering were in fact
influences, and that such regulatory networks should be called influence
networks. At that time, we were more interested in biochemical reaction networks
for which Thomas's conditions generally provide no information.
That was the starting point of an adventure that led
the third author to refine Thomas's conditions for reaction networks,
and the first author to implement and successfully apply them in large scale,
showing the richness of Ren\'e Thomas's intuitions across so many decades of
active research in mathematical biology.
}

\end{document}